\documentclass{article}
\usepackage{ijcai25}

\usepackage{times}
\usepackage{soul}
\usepackage{url}
\usepackage[hidelinks]{hyperref}
\usepackage[utf8]{inputenc}
\usepackage[small]{caption}
\usepackage{graphicx}
\usepackage{amsmath}
\usepackage{amsthm}
\usepackage{booktabs}
\usepackage{algorithm}
\usepackage[switch]{lineno}

\usepackage{amsmath,amsfonts,amssymb,amsthm}
\usepackage{boxedminipage}

\usepackage{amsmath, amssymb, algorithmicx, algpseudocode}

\usepackage{mathtools}
\usepackage{algorithm}

\usepackage{framed}
\usepackage{thmtools}
\usepackage[backgroundcolor=gray!10,textsize=tiny]{todonotes}
\usetikzlibrary{calc,math,patterns,shapes,backgrounds}
\usepackage{dsfont}
\usepackage{xifthen}
\usepackage{tabularx}
\usepackage{tikz}
\usepackage{paralist}
\usepackage[no abbrev]{cleveref}

\newcommand{\citet}[1]{\citeauthor{#1}~\shortcite{#1}}
\newcommand{\citep}[1]{\cite{#1}}

\newtheorem{theorem}{Theorem}
\newtheorem{lemma}[theorem]{Lemma}

\newtheorem{observation}[theorem]{Observation}
\newtheorem{proposition}[theorem]{Proposition}
\newlength{\RoundedBoxWidth}

\crefname{theorem}{Theorem}{Theorems}
\crefname{observation}{Observation}{Observations}
\crefname{proposition}{Proposition}{Propositions}
\crefname{corollary}{Corollary}{Corollaries}
\crefname{lemma}{Lemma}{Lemmas}
\crefname{claim}{Claim}{Claims}
\crefname{definition}{Definition}{Definitions}
\crefname{section}{Section}{Sections}
\crefname{figure}{Figure}{Figures}

\Crefname{theorem}{Thm}{Thms}
\Crefname{observation}{Obs}{Obs}
\Crefname{proposition}{Prop}{Props}
\Crefname{corollary}{Cor}{Cors}

\newsavebox{\GrayRoundedBox}
\newenvironment{GrayBox}[1]%
   {\setlength{\RoundedBoxWidth}{.92\linewidth}
    \def\boxheading{#1}
    \begin{lrbox}{\GrayRoundedBox}
       \begin{minipage}{\RoundedBoxWidth}}%
   {   \end{minipage}
    \end{lrbox}
    \begin{center}
    \begin{tikzpicture}%
       \node(Text)[draw=black!20,fill=white,rounded corners,%
             inner sep=2ex,text width=\RoundedBoxWidth]%
             {\usebox{\GrayRoundedBox}};
        \coordinate(x) at (current bounding box.north west);
        \node [draw=white,rectangle,inner sep=3pt,anchor=north west,fill=white] 
        at ($(x)+(6pt,.75em)$) {\boxheading};
    \end{tikzpicture}
    \end{center}}     

\newcommand{\defproblema}[3]{
	\begin{GrayBox}{\textsc{#1}}
		\begin{asparadesc}
			\item[Input:] #2
			\item[Question:] #3
		\end{asparadesc}
	\end{GrayBox}
}

\newcommand{\Wone}{\ensuremath{\mathrm{W}[1]}}
\newcommand{\NP}{\ensuremath{\mathrm{NP}}}

\newcommand{\fair}{\textsc{Envy Elimination by Adding Goods}}
\newcommand{\ufair}{\textsc{Envy Elimination by Adding Goods}}

\newcommand{\N}{\ensuremath{\mathds{N}}}

\DeclareMathOperator{\supply}{\#}

\newcommand{\finR}{\ensuremath{R_{<\infty}}}

\makeatletter
\newcommand{\pushright}[1]{\ifmeasuring@#1\else\omit\hfill$\displaystyle#1$\fi\ignorespaces}
\makeatother

\newcommand{\yes}{\textnormal{\texttt{yes}}}
\newcommand{\no}{\textnormal{\texttt{no}}}

\title{How to Resolve Envy by Adding Goods}

\author{
	Matthias Bentert$^1$
	\and
	Robert Bredereck$^2$
	\and
	Eva Deltl$^2$
	\and
	Pallavi Jain$^3$
	\and
	Leon Kellerhals$^2$\\
	\affiliations
	$^1$University of Bergen, Norway \\
	$^2$TU Clausthal, Germany \\
	$^3$Indian Institute of Technology Jodhpur, India\\
	\emails
	matthias.bentert@uib.no,
	\{robert.bredereck,eva.deltl,leon.kellerhals\}@tu-clausthal.de,
	pallavijain.t.cms@gmail.com}

\date{}

\begin{document}

\maketitle
\begin{abstract}
We consider the problem of resolving the envy of a given initial allocation by adding elements from a pool of goods.
We give a characterization of the instances where envy can be resolved by adding an arbitrary number of copies of the items in the pool.
From this characterization, we derive a polynomial-time algorithm returning a respective solution if it exists.
If the number of copies or the total number of added items are bounded, the problem becomes computationally intractable
even in various restricted cases.
We perform a parameterized complexity analysis, focusing on the number of agents and the pool size as parameters.
Notably, although not every instance admits an envy-free solution, our approach allows us to efficiently determine, in polynomial time, whether a solution exists—an aspect that is both theoretically interesting and far from trivial.

\end{abstract}


\section{Introduction}
\label{sec:intro}


Fair allocation of indivisible goods to agents is a fundamental problem in resource allocation~\cite{DBLP:journals/ai/AmanatidisABFLMVW23,DBLP:journals/jair/LiuLSW24,DBLP:journals/corr/abs-2307-10985,DBLP:conf/ijcai/0004R23}
and has a plethora of applications such as dorm room allocations, donations by a food bank, inheritance matters, divorce disputes, or job allocations.
While most work assumes that the allocation starts with a blank sheet, i.e., all items are initially unassigned and the agents do not envy each other,
there are many scenarios where some items are already allocated.
For example, in inheritance matters, the testament may likely allocate some of the heritable goods in an unfair way.

Recent work provides a toolbox of approaches to increase the fairness of such unfair allocations,
such as deleting or donating some of the allocated goods~\cite{DBLP:journals/algorithmica/DornHS21,DBLP:journals/siamcomp/ChaudhuryKMS21,boehmer2024multivariate},
sharing goods~\cite{DBLP:journals/ior/SandomirskiyS22,DBLP:journals/jair/BredereckKLNS23}, reallocating them~\cite{AZIZ20191,10.5555/3171642.3171674}, or providing subsidies~\cite{DBLP:conf/sagt/HalpernS19,DBLP:conf/sigecom/BrustleDNSV20,DBLP:conf/ijcai/BarmanKNS22}.
Naturally, these approaches to increase fairness cannot be a good match for every real-world scenario.
Donating or sharing goods is not always possible:
The goods could have lost their value already (think of food donations by a community center),
the current allocation may need to abide some constraints such as a testament,
or the agents may become attached to their assigned goods.
Clearly, subsidies are also not amenable in every scenario, as they may fail to address urgent non-monetary needs such as shelter, food, or medical care.
Sometimes, certain approaches, such as reallocation, sharing, or subsidies,
are impossible for legal reasons. 

We extend the toolbox with an approach for scenarios where the initial allocation is fixed
and the goal is to alleviate unfairness by allocating items from a pool of additional goods.
The feature of adding items that sets it apart from the above approaches is that, theoretically, there is no natural limit on the budget:
While one cannot delete/reallocate/share more than all items of the initial allocation, there is no bound on the number of items one can add.
Of course, if one could add from the infinite pool of all possible resources, the problem would become trivial.
Thus, we assume to be given a finite set of resources with possibly unlimited supply.



While true infinity of course is unattainable in reality, scenarios where resources are plentiful
or practically sufficient provide a fertile ground for applications of our theoretical results.
Consider, for instance, a community center distributing leftover charity items.
In such a setting, the center might possess a large stock of various goods that can be reallocated
with minimal constraint on quantities.
This abundance allows for flexible distributions that can address disparities and highlights a scenario where adding items is feasible, whereas redistributing money is not.
Similarly, in some settings, offering cash may be ineffective or even harmful. Recipients might purchase unsuitable items, and 
non-monetary needs 
may not be immediately met. Here, directly adding appropriate goods or services is more impactful than offering cash.
Another instance of practical abundance 
involves the use of vouchers or gift cards, which can be thought of as having an infinite supply from the perspective of the allocating authority.
Here, the challenge of achieving fairness can be substantially alleviated by issuing as many 
as needed to satisfy fairness, thereby achieving a more balanced distribution of intangible resources.
Additionally, in the context of office liquidations, there may exist vast amounts of leftover resources such as
laptops, chairs, and other equipments. 

We study the computational complexity of alleviating unfairness by adding items from a pool of goods to an initial allocation.
Herein, we focus solely on envy-freeness and additive, non-negative valuations.


\paragraph{Related work.}
When envy-freeness cannot be achieved, one approach is to compensate agents with a divisible resource (e.g., money) to eliminate envy \cite{haake2002bidding}.
The study of fair division with subsidies or money transfers \cite{DBLP:conf/sagt/HalpernS19,DBLP:conf/aaai/000121} has focused on finding allocations for which minimal subsidies will result in envy-freeness.
Instead of adding one divisible resource, our model considers achieving envy-freeness by adding a given set of additional indivisible items.

The idea of controlling fair division scenarios by adding or deleting items, in order to arrive at an instance which can be allocated fairly, has been explored in various contexts.
\citet{aziz2016control} consider the problem of adding/deleting/replacing few items such that there is an envy-free complete allocation with ordinal valuations.
This is similar to our setting when our initial allocation is empty; their NP-hardness results carry over to our problem with ordinal valuations (but not with cardinal, additive valuations).

While we focus on achieving envy-freeness through the addition of resources, previous studies have primarily explored fairness by considering item deletions.
For example, \citet{DBLP:journals/algorithmica/DornHS21} analyzed the complexity of achieving proportional allocations by removing items in settings with ordinal preferences.
Similarly, \citet{boehmer2024multivariate} investigated the complexity of ensuring envy-free (and EF1) allocations through item donations, when agents have additive utility-based valuations. Unlike these approaches, our work does not allow modifications to the initial allocation.

Closest to ours is the recent work by \citet{hv2024fairefficientcompletionindivisible}
who looked at scenarios where specific resources have been pre-assigned to particular agents
and the goal is to complete the allocation by allocating all remaining resources such that the overall allocation is fair and efficient.
The key difference to our work is that our pool of additional goods does not need to be fully allocated and that our focus lies on variants where the supply is unbounded.

\paragraph{Our contributions.}
We study the computational complexity of \fair{} (defined in \cref{sec:prelims}),
where the task is to extend an initial unalterable allocation by adding items from a pool of additional goods to the agents such that the extended allocation becomes envy-free.
We focus on the setting where the agents have additive, non-negative valuations.

A major difference to settings where one can delete or reallocate items from the initial allocation is that there is no natural bound on the number of additionally allocated items.
For example, the problem of deleting an unlimited amount of items from an initial allocation is trivial: empty allocations are envy-free.
As it turns out, this is not the case when we add an unlimited amount of items which all have infinite supply:
Consider the following example with two agents with identical valuations.
The initial allocation gives an item of value one to the first agent, and we may add any number of copies of an item with value two.

While such observations often pave the way for an NP-hardness proof, our main technical contribution shows that whether one can resolve envy by adding any number of copies of a given set of items can be decided in polynomial time (\cref{sec:poly}).
The main challenge herein are agents whose valuations over the additional items are identical (or proportional).
As a special case, the envy between two agents can only be resolved by an extension if there is an integer divisible by the greatest common divisor within a certain ranges quantified by the agents' envy gaps.
To resolve the envy between multiple agents, we verify whether this property can be simultaneously fulfilled by an integer linear program (ILP) derived from the above ranges.
Polynomial-time solvability of the ILP follows from the constraint matrix being totally unimodular.


We next studied how bounding the supply of some items (\cref{sec:bounded-supply}) or size (also called budget) of the extension (\cref{sec:bounded-budget}) affects the computational complexity of our problem.
Motivated by computational intractability even when there are only three additional items with finite supply,
we initiate a parameterized complexity analysis of our problem, showing two other parameterized hardness results,
but also fixed-parameter tractability by the sum of item supplies if the budget is bounded and by the parameter combination number of agents plus number of resources.
We refer to \cref{fig:overview} for an overview of our results.
\tikzset{
	textpara/.style = {
		align = center,
		font = \footnotesize,
	},
	para/.style = {
		textpara,
		draw,
		rectangle,
		rounded corners = 1mm,
	},
	paranp/.style = {
		para,
		fill=red!30,
	},
	paraxp/.style = {
		para,
		fill=orange!30,
	},
	paranopk/.style = {
		para,
		fill=yellow!50,
	},
	parapk/.style = {
		para,
		fill=green!30,
	},
	parawhat/.style = {
		para,
		fill=black!20,
	},
	parapoly/.style = {
		para,
		fill=blue!30!white,
	},
}
\begin{figure}
	\centering
	\begin{tikzpicture}
		\node[parapoly, minimum width=10em] (poly) at (0, 2.9) {\textbf{Polynomial-time solvable} (\Cref{thm:poly})};
		\draw[semithick,dashed,black!70,rounded corners] (-4.25,3.3) rectangle (4.25, 2.15);
		\node at (0,2.4) [font=\small] {\emph{all supply unbounded}};

		\node[paranopk] (k) at (-3, 0) {$k$\\ (Thms \ref{thm:hardness},\,\ref{thm:fpt-r-bounded})\\\textbf{XP, \Wone-hard}};
		\node[paranp]   (r) at (0, -.7) {$|R|$\\ (\Cref{prop:np-h-r})\\\textbf{para-NP-hard}};
		\node[paraxp]   (n) at (3.25, 0) {$|A|$\\ (\Cref{prop:ufair-wone-a})\\\textbf{\Wone-hard}};

	\node[parapk]   (p) at (-2, 1.4) {$p$\\ (\Cref{thm:fpt-r-bounded})\\\textbf{FPT}};
		\node[parapk]  (nr) at (2, 1.4) {$|A| + |R|$\\ (\Cref{prop:ILP})\\\textbf{FPT}};

		\draw[semithick] (k) -- (p) -- (r) -- (nr) -- (n);

		\draw[semithick,dashed,black!70,rounded corners] (-4.25,2) rectangle (-1.25, -1.4);
		\node at (-2.75,-1.15) [font=\small] {\emph{bounded budget}};
	\end{tikzpicture}
	\caption{
		Overview over our results for \fair{} with different parameterizations.
        Results in dashed boxes hold only for instances with the respective restriction.
		Here, $A$ denotes the set of agents, $R$ the set of additional items, $k$ the upper bound on the number of items, and $p$ the sum of the supplies.
		An edge between two parameter boxes implies that the upper parameter upper-bounds the lower parameter.
	}\label{fig:overview}
\end{figure}
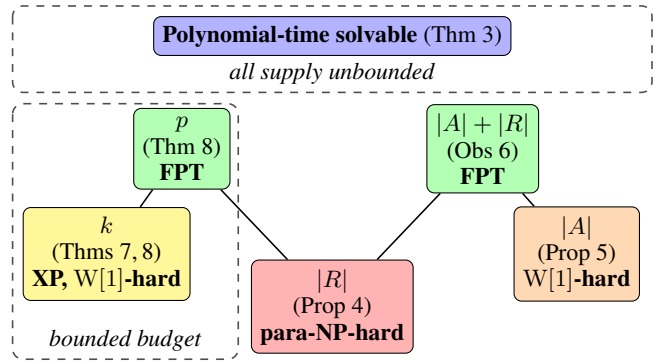

\section{Preliminaries}
\label{sec:prelims}

For~$n \in \N$, we denote by~$[n]$ the set~$\{1,2,\ldots,n\}$.

Let $A = \{a_1, \dots, a_n\}$ be the set of \emph{agents}.
We have two sets of items:
The (multi-)set $P$ of \emph{initial items} and the set $R$ of \emph{additional items}.
Moreover, each item $r \in R$ has a supply $\supply(r) \in \N \cup \{\infty\}$ and each agent $a \in A$ has a \emph{valuation} $v_a \colon P \cup R \to \N$.
Throughout this work, we assume valuation functions to be additive and define $v_a(S) \coloneq \sum_{r \in S} v_a(r)$ for all $S \subseteq P \cup R$.
We call a valuation $v_a$ \emph{binary} if $v_a(r) \in \{0,1\}$, and say that $a$ approves item $r$ if $v_a(r) = 1$.
Two valuations $v_a$ and $v_{a'}$ are \emph{identical} if $v_a = v_{a'}$
and \emph{proportional} if there is an $\alpha \in \mathds{Q}$ such that $v_a(r) = \alpha v_{a'}(r)$ for all $r \in R$.
We then write $v_a = \alpha v_{a'}$ for short.

In our work, we are given an \emph{initial allocation} $\sigma \colon A \to 2^P$,
which allocates sets of initial items to agents.
The set $\sigma(a)$ is called the \emph{initial bundle of $a$} and is disjoint from the initial bundle of any other agent.

We look for an \emph{extension} $\rho \colon A \times R \to \N$ which specifies the number of copies of item $r$ to allocate to agent~$a$.
We assert $\sum_{a \in A} \rho(a,r) \le \supply(r)$ for each $r \in R$.
The \emph{size} or \emph{budget} of $\rho$ is the number $\sum_{a \in A} \sum_{r \in R} \rho(a, r)$ of allocated items.
For agents $a, a' \in A$, we define $v_a(\rho,a') \coloneqq \sum_{r \in R} \rho(a',r) \cdot v_a(r)$.

An \emph{extended allocation} $(\sigma, \rho)$ consists of an initial allocation $\sigma$ and an extension $\rho$.
For two agents $a, a' \in A$ and an extended allocation $(\sigma,\rho)$, we say that $a$ \emph{envies} $a'$ if
\[v_{a}(\sigma(a)) + v_a(\rho,a) < v_{a}(\sigma(a'))+v_{a}(\rho,a').\]
We say that~$a$ initially envies~$a'$ if~$v_a(\sigma(a)) < v_a(\sigma(a'))$.
We say that $a$ is (initially) \emph{envious} if it (initially) envies some other agent.
The \emph{envy graph} $\mathcal{G}_{\sigma,\rho}$ of $(\sigma, \rho)$
has a vertex for each agent and a directed edge $(a, a')$ (also called \emph{envy edge}) if $a$ envies $a'$.
For any pair $a, a'$ of agents, the \emph{envy gap} of $a$ and~$a'$ with respect to~$(\sigma,\rho)$ is $\gamma_\rho(a,a') = v_{a}(\sigma(a'))+ v_a(\rho, a') - (v_{a}(\sigma(a)) + v_a(\rho, a))$.
We call an extended allocation $(\sigma, \rho)$ \emph{envy-free} if $\mathcal G_{\sigma, \rho}$ is edgeless.
We then also say that $\rho$ \emph{resolves envy}.
The initial envy gap of~$a$ and~$a'$ is~$\gamma(a,a') = v_a(\sigma(a'))-v_a(\sigma(a))$.
We mention that the (initial) envy gap can be both positive and negative (as well as 0).

In this work, we consider the computational problem of finding an envy-resolving extension,
defined as follows.

\defproblema{\fair}{
	Sets $P, R$ of items, a supply $\supply(r) \in \N \cup \{\infty\}$ for each $r \in R$,
	a set $A$ of agents, each with a valuation function $v_a \colon P \cup R \to \N$,
 an initial allocation $\sigma$, and an integer $k \in \N$.
}{
	Is there a size-$k$ envy-resolving extension $\rho$?
}

We point out that the main part of this work is on the case where the budget $k$ is unrestricted, i.e.,
there is no restriction on the number of items the extension can allocate.

We use standard notation from parameterized algorithmics \cite{bluebook}.
A problem parameterized with some integer $k$ is called
\begin{inparaenum}[(1)]
\item fixed-parameter tractable (FPT), if it can be solved in $f(k) \cdot |I|^{\mathcal O(1)}$ time, and
\item in XP, if it can be solved in $|I|^{f(k)}$ time,
\end{inparaenum}
where $f$ is a computable function only depending on $k$ and $|I|$ is the instance size.
One can show that a parameterized problem is likely not FPT by proving it to be \Wone-hard (using parameterized reductions, which generalize standard polynomial-time many-one reductions).
Moreover, if a parameterized problem is NP-hard even for constant parameter values, we call it para-NP-hard.
In our running time analysis, we assume that basic arithmetic operations over any numbers in the input (or of similar size) can be performed in constant time.

\section{When Supply and Budget are Unbounded}
\label{sec:poly}

We first look at the case where each item $r \in R$ has infinite supply.
We show that this case is polynomial-time solvable.
We first show how to decide whether the envy between two agents can be resolved.
Afterwards, we generalize the result to more than two agents.

We distinguish between agents with non-proportional and proportional valuations.
In the first case, envy can always be overcome as shown next.
Observe that valuations can only be non-proportional if there are at least two items in $R$.

\begin{lemma}\label{lem:diff}
	Let $a, a'$ be two agents with non-proportional valuations.
	If each item in $R$ has infinite supply,
	then there exists an extension $\rho$ that resolves envy between $a$ and $a'$.
	Moreover, $\rho$ can be computed in $\mathcal{O}(m)$ time, where $m = |R|$.
\end{lemma}

\begin{proof}
	Non-proportionality implies the existence of \( {r_1, r_2 \in R} \) with differing ratios of valuations between agents.
	For brevity, let
\( v_{a}(r_1) = x \), \( v_{a}(r_2) = c \), \( v_{a'}(r_1) = y \), and~\( {v_{a'}(r_2) = d} \).
Furthermore, assume without loss of generality that \( r_1 \) is more valuable relative to \( r_2 \) for agent \( a \) compared to agent \( a' \), i.e., 
\[
\frac{x}{y} > \frac{c}{d},
\]
which is equivalent to $cy < xd$.
As all of these values are integers, we have $xd - cy \ge 1$.
We now show how we can leverage the differing valuation ratios to eliminate envy by careful distribution of these items.

If initially $a$ envies $a'$,
consider the extension $\rho$ with
$\rho(a, r_1) = d$ and 
$\rho(a', r_2) = y$.
Then
\begin{align*}
    v_{a}(\rho,a) &= d \cdot x > y \cdot c = v_{a}(\rho,a') \quad\text{and}\\ v_{a'}(\rho,a) &= d\cdot y = y \cdot d = v_{a'}(\rho,a'),
\end{align*}
that is, the gap $\gamma(a, a')$ decreases by~$xd - cy \geq 1$ while the valuation of $a'$ regarding the bundles of~$a$ and~$a'$ does not change as both values increase by~$dy$.
Thus, allocating $\gamma(a,a')$ copies of $\rho$ resolves the envy of $a$ to $a'$ without creating new envy from $a'$ to $a$.
If initially~$a'$ envied~$a$, then we can do the same with the roles of~$a$ and~$a'$ reversed by
allocating $c \cdot \gamma(a', a)$ copies of $r_1$ to $a$
and $x \cdot \gamma(a', a)$ copies of $r_2$ to $a'$.


Note that finding $r_1,r_2 \in R$ with different valuation ratios takes $\mathcal O(m)$ steps as we can start with any object and then all remaining objects either have the same valuation ratio or at least one differs.
Hence, we can compute an extension that resolves envy between~$a$ and~$a'$ in $\mathcal{O}(m)$ time.
\end{proof}

We next turn to the case with proportional valuations.
We first prove a necessary and sufficient condition for resolving envy between two agents $a$ and $a'$.
We mention that the valuation of initial items does not need to be proportional.

\begin{lemma}\label{lem:gcd}
	Let $a, a'$ be two agents with $v_a(r) = \alpha v_{a'}(r)$ for all~$r \in R$ and some~$\alpha$.
    Suppose that $a$ envies $a'$ and let $\gamma(a,a')$ and $\gamma(a',a)$ be the envy gaps of $a$ and~$a'$ and of~$a'$ and~$a$, respectively.
    If each item in $R =  \{r_1,\ldots,r_m\}$ has infinite supply,
    then there exists an extension $\rho$ that resolves the envy between $a$ and $a'$
    if and only if
    there exists an integer $\gamma(a,a') \le T \le - \alpha \gamma(a',a)$ that is divisible by
    $\gcd(v_a(r_1), \dots, v_a(r_m))$.
    Moreover, $\rho$ can be computed in $\mathcal{O}(m \log W)$ time, where $W$ is the largest item valuation.
\end{lemma}
\begin{proof}
	For brevity we let $b_i \coloneqq \rho(a, r_i)$ and $b'_i \coloneqq \rho(a', r_i)$ for each $i \in [m]$.
	Then $\rho$ is envy-resolving if and only if
	\begin{align*}
		&v_a(\sigma(a)) + \sum_{i=1}^m b_i v_a(r_i) \ge v_a(\sigma(a')) + \sum_{i=1}^m b'_i v_a(r_i)\\
		\text{and } & v_{a'}(\sigma(a')) + \sum_{i=1}^m b'_i v_{a'}(r_i) \ge v_{a'}(\sigma(a)) + \sum_{i=1}^m b_i v_{a'}(r_i).
	\end{align*}
	These inequalities can be reformulated as
	\begin{align*}
		&\sum_{i=1}^m (b_i - b'_i) v_a(r_i) \ge v_a(\sigma(a')) - v_a(\sigma(a)) = \gamma(a,a'), \text{ and}\\
		& \sum_{i=1}^m (b_i - b'_i) v_{a'}(r_i) \le v_{a'}(\sigma(a') - v_{a'}(\sigma(a)) = -\gamma(a',a).
	\end{align*}
	As $v_a(r) = \alpha v_{a'}(r)$ for each $r \in R$,
	$\rho$ is envy-resolving if and only if
	\begin{align*}
		\gamma(a,a') \le& \textstyle\sum_{i=1}^m (b_i - b'_i) v_a(r_i)
		       \\=&   \textstyle\sum_{i=1}^m (b_i - b'_i) \alpha v_{a'}(r_i) \le -\alpha \gamma(a',a).
	\end{align*}
Next, if $d \coloneqq \gcd(v_a(r_1), \dots, v_a(r_m))$, then there are $k_i \in \N$ such that $v_a(r_i) = d \cdot k_i$.
	Thus,
	\begin{align*}
		\textstyle\sum_{i=1}^m (b_i - b'_i) v_a(r_i) = d \sum_{i=1}^m (b_i - b'_i)k_i,
	\end{align*}
	and $\rho$ is envy-resolving if and only if there is an integer \[T = d \textstyle\sum_{i=1}^m (b_i - b'_i) k_i\] with $\gamma(a, a') \le T \le -\alpha \gamma(a', a)$.
	Clearly, $T$ is divisible by $d$.

We now show how such an extension can be computed if it exists.
By Bézout's Lemma \cite{bezout1779}, the equation $d \coloneqq \gcd(v_a(r_1), v_a(r_2), \dots, v_a(r_m))$ implies the existence of integers \( c_1, c_2, \dots, c_m \) with
\(
	d = \sum_{i=1}^m c_i v_a(r_i).
\) 

\newcommand{\eea}{\textsc{Extended Euclidean Algorithm}}
The computation of $d$ and the coefficients $c_i$ can be carried out iteratively using the \eea \ \cite{knuth1997art}. 
As $d$ divides $T$, there is a $q \in \N$ such that $T = q \cdot d$;
therefore
\[
	T = q \cdot \textstyle \sum_{i=1}^m c_i v_a(r_i) = \textstyle\sum_{i=1}^m (q \cdot c_i) \cdot v_a(r_i).
\]
We then decompose each $q \cdot c_i$ into positive and negative parts
\[
    b_i = \max\{0, q \cdot c_i\}, \quad b'_i = \max\{0, -q \cdot c_i\}.
\]
Then $b_i - b'_i = q \cdot c_i$, and therefore $\sum_{i=1}^m (b_i - b'_i) \cdot v_a(r_i) = T$.

Once coefficients $b_i$ and $b'_i$ are computed, we define $\rho$ to be an extension, such that $\rho(a,r_i) = b_i$ and $\rho(a',r_i) = b'_i$, for all $r_i \in R$. 
As $\gamma(a,a') \leq T \leq -\alpha\gamma(a',a)$, we get
\begin{align*}
&v_a(\sigma(a)) +\textstyle\sum_{i=1}^m b_i v_{a}(r_i) \\={} &v_a(\sigma(a)) + \textstyle\sum_{i=1}^m b'_i v_{a}(r_i) + T \\\geq {}&v_a(\sigma(a)) + \textstyle\sum_{i=1}^m b'_i v_{a}(r_i) + \gamma(a,a') \\={} &v_a(\sigma(a')) + \textstyle\sum_{i=1}^m b'_i v_{a}(r_i).
\end{align*}
This confirms that agent $a$ does not envy agent $a'$  under the extended allocation. Similarly, since $v_a(r) = \alpha v_{a'}(r)$,
\begin{align*}
      & v_{a'}(\sigma(a')) +      \textstyle\sum_{i=1}^m b'_i v_{a'}(r_i)\\
    ={} & v_{a'}(\sigma(a')) +    \textstyle\sum_{i=1}^m b_i v_{a'}(r_i) - \frac{T}{\alpha} \\
    \geq{} & v_{a'}(\sigma(a')) + \textstyle\sum_{i=1}^m b_i v_{a'}(r_i) + \gamma(a',a)\\
    ={} & v_{a'}(\sigma(a)) +     \textstyle\sum_{i=1}^m b_i v_{a'}(r_i).
\end{align*}
Thus, agent $a'$ does not envy agent $a$, ensuring that neither agent envies the other in the extended allocation.

The \eea{} runs in $\mathcal{O}(m\cdot\log W)$ time where $W = \max_{a \in A, r \in R} v_a(r)$.
Scaling the coefficients and decomposing them into $b_i$ and $b'_i$ is linear in $m$. Thus, the overall running time is in~$\mathcal{O}(m\log W)$.
\end{proof}

We next generalize the above to more than two agents and show the main result of this section.

\begin{theorem}\label{thm:poly}
    \ufair{} is polynomial-time solvable when $\supply(r) = \infty$ for all $r\in R$ and $k$ is unbounded.
\end{theorem}

\begin{proof}

Let $\sigma$ be an initial allocation.  We start with an empty extension $\rho$. Let $\mathcal{G}_{\sigma,\rho}$ denote the envy graph. Note that, as shown in \cref{lem:gcd}, finding an envy-resolving extension is not always possible.

Our algorithm proceeds in two phases. The first phase addresses envy among agents which have proportional valuations of all items in~$R$, while the second phase eliminates envy between agents with non-proportional valuations.

To implement the first phase, we utilize the fact that proportional valuations form an equivalence relation. This allows us to partition 
$A$ into equivalence classes $A_1 \uplus \dots \uplus A_t$, where $a,a' \in A_i$ if and only if $a$ and $a'$ have proportional valuations over~$R$. We begin by trying to resolve envy edges within each equivalence class.

Once this is completed, we move on to the second phase, where we eliminate envy between agents from different equivalence classes.

\paragraph{Phase 1: Eliminating envy between two agents with proportional valuations.}

Let~$A_i$ be an equivalence class and recall that~$R=\{r_1,r_2,\ldots,r_m\}$.
First, we compute a normalized valuation $v'$ as follows.
For each~$a \in A_i$, we compute~$d_a = \gcd(v_a(r_1),v_a(r_2),\ldots,v_a(r_m))$ and set~$v'_a(r_i) = \frac{v_a(r)}{d_a}$ for all~$r \in P \cup R$.
Note that~$v'_a(r)$ is an integer for all $r \in R$ but not necessarily for~$r \in P$.
Moreover, it holds that~$\gcd(v'_a(r_1),v'_a(r_1),\ldots,v'_a(r_m)) = 1$ and an extension~$\rho$ resolves envy with respect to valuations~$v'_a$ if and only if it resolves envy with respect to the original valuation~$v_a$.
Note that the envy gap~$\gamma(a,a')$ of two agents~$a,a' \in A_i$ with respect to~$v'$ might not be integral.
However, we can round all these values up to the nearest integer as valuations can only go up in integer steps as~$v'_a(r)$ is an integer for each~$r \in R$.
Slightly abusing notation, we will refer to~$\gamma'(a,a')$ as the rounded value.

We next show that~$v'_a(r)  = v'_{a'}(r)$ for all~$a,a' \in A_i$ and all~$r \in R$.
Since~$a,a' \in A_i$, it holds that~$v_a(r) = \alpha v_{a'}(r)$ for all~$r \in R$.
Moreover, since~$v_a(r)$ and~$v_{a'}(r)$ are non-negative integers, it holds that~$\alpha= \frac{p}{q}$ for two positive coprime integers~$p$ and~$q$.
We claim that $p=q=1$.
If~$p \neq 1$, then~$v'_a(r)$ is divisible by~$p$ for all~$r \in R$, contradicting the fact that~$\gcd(v'_a(r_1),v'_a(r_1),\ldots,v'_a(r_m)) = 1$.
Similarly, if $q \neq 1$, then~$v'_{a'}(r)$ are divisible by~$q$ for each~$r \in R$, another contradiction.
Hence~$\alpha = 1$ and~$v'_a(r) = v'_{a'}(r)$ for all~$r \in R$.

We next build an integer linear program (ILP) that encodes whether a solution exists.
We start with a variable~$x_a$ for each agent~$a \in A_i$ that roughly encodes how much utility agent~$a$ should get in a solution (with respect to the valuation~$v'_a$).
Note that since~$v'_a(r) = v'_{a'}(r)$ for all~$r \in R$, the meaning of variable~$x_a$ is universal for all agents in~$A_i$.
For each pair~$a,a' \in A_i$ of agents, we add the constraint~$x_a - x_{a'} \geq \gamma'(a,a')$.
Note that~$\gamma'(a,a')$ is a constant that can be precomputed in polynomial time.
This constraint encodes that if~$a$ gets an additional utility of~$x_a$ and~$a'$ gets an additional utility of~$x_{a'}$, then~$a$ does not envy~$a'$.
If there is an extension that resolves all envy within~$A_i$, then the constructed ILP has a solution.
It remains to show how to construct an extension from a solution to the ILP and to discuss why a solution to the constructed ILP can be found in polynomial time.

Towards the former, note that if a solution exists, then it remains a solution if all agents in~$A_i$ get~$c$ additional utility for any~$c \in \N$.
Moreover, we use the fact that $\gcd(v'_a(r_1),v'_a(r_2),\ldots,v'_a(r_m))=1$ for all~$a\in A_i$.
This implies that there are two multisets~$X,Y$ of items such that~$v'_a(X) = v'_a(Y) + 1$.
If a solution is found, then for each agent~$a \in A_i$, we do the following.
For each~$a \in A_i$, let~$z_a$ be the value assigned to variable~$x_a$ in the solution.
We give~$z_a$ times the bundle~$X$ to agent~$a$ and~$z_a$ times the bundle~$Y$ to all other agents in~$A_i$.
Note that the valuation of all bundles except for~$a$ increases by~$z_a \cdot v'_a(Y)$ and the valuation of the bundle of~$a$ increases by~$z_a \cdot v'_a(X) = z_a \cdot v'_a(Y) + z_a$.
That is, the valuation of~$a$'s bundle increases by precisely~$z_a$ compared to the bundles of all other agents in~$A_i$.
After doing the above for all agents in~$A_i$, it holds that the difference in valuation of the bundles of~$a$ and~$a'$ changed by~$z_a-z_{a'}$, which is precisely what the solution to the ILP required.
Since all of the above can be computed in polynomial time, we can compute an extension resolving envy given a solution to the constructed ILP.

To conclude the first phase, we note that the constructed ILP encodes a totally unimodular matrix as each constraint contains exactly two variables and one has coefficient 1 and the other has coefficient -1 \cite{Tam76}.
Since all constants in the ILP are integral, it holds that the corresponding polyhedron only has integer vertices and the ILP can be solved in polynomial time \cite{HK56} (or it can be concluded that no (integral) solution exists).

\paragraph{Phase 2: Eliminating envy between two agents with non-proportional valuations.}

Once the first phase has been applied to every equivalence class $A_i \subseteq A$, any remaining envy must necessarily exist between agents with non-proportional valuations.
Let $E_{\mathrm{diff}}$ denote the set of remaining envy edges between agents with non-proportional valuations. Next, we eliminate envy between two agents with non-proportional valuations without creating new envy similar to what we did in \cref{lem:diff}.

Let $a, a'$ be two agents with non-proportional valuations such that $a$ envies $a'$. By \cref{lem:diff}, we know that it is always possible to compute an extension $\rho'$ to eliminate this envy.
Let $x, x'$ be the number of items $r, r'$, respectively, added in $\rho'$ to resolve the envy between $a$ and $a'$. 

We update $\rho$ such that for each agent $a^* \in A$ we set
    $\rho(a^*, r) = \rho(a^*, r) + x$  if $x \cdot v_{a^*}(r) \geq x' \cdot v_{a^*}(r')$ and $\rho(a^*, r') = \rho(a^*, r') + x'$ otherwise.
For an agent $a^*$, assume w.l.o.g. that $x \cdot v_{a^*}(r) \geq x' \cdot v_{a^*}(r')$. If agent $a^*$ did not envy an agent $a$ before the allocation, no matter which item set $a$ gets, no new envy will be created. This follows from the fact that $v_{a^*}(\sigma(a^*)) \geq v_{a^*}(\sigma(a))$ implies \begin{align*}
    v_{a^*}(\sigma(a^*)) + x \cdot v_{a^*}(r) &\geq v_{a^*}(\sigma(a)) + x \cdot v_{a^*}(r)\\ &\geq v_{a^*}(\sigma(a)) + x' \cdot v_{a^*}(r').
\end{align*}

As no new envy edges are created, we can apply this step until $E_{\mathrm{diff}}=\emptyset$, implying that we have reached an envy-free extended allocation of the items.

\paragraph{Running time.}

As shown above, phase one can be solved in polynomial time for each equivalence class~$A_i$.
Since the number of equivalence classes is at most~$n$ and they can be computed in polynomial time, the first phase takes polynomial time overall.

At the beginning of phase $2$ there can be at most $n^2$ envy edges. For non-proportional envy edges, \cref{lem:diff} guarantees that no new envy edges are created during their elimination. Consequently, this step is executed at most $n^2$ times.
As resolving one conflict edges takes~$\mathcal{O}(m)$ time, the total running time for the second phase is in \(\mathcal{O}(n^2 \cdot m).\)
This concludes the proof of \cref{thm:poly}.
\end{proof}

\section{When some Supply is Bounded}
\label{sec:bounded-supply}

We next focus on the case where we still have no restrictions on the budget, but some items may have finite supply.
This makes finding an envy-resolving extension computationally hard,
even if we have only three different items with finite supply.
Also, as we will see, there is little hope for fixed-parameter tractability for the number of agents alone.
However, when parameterizing by the number of agents plus the number of (different) items, the problem becomes fixed-parameter tractable.

We first show the two hardness results.

\begin{proposition}
	\label{prop:np-h-r}
	\ufair{} is \NP-hard even if there are three additional items,
	if the budget is unbounded,
	and the valuations are binary.
\end{proposition}

\begin{proof}
	We reduce from the \textsc{Clique} problem, where the task is to decide
	whether a given graph $G = (V, E)$ contains a complete subgraph with a given number $\ell$ of vertices.
	We have an agent $a_v$ for each vertex $v \in V$, an agent $a_e$ for each edge $e \in E$, and an extra agent $b$.
	Initially, $b$ holds an item that is approved by $b$ and each edge agent $a_e$,
	each vertex agent $a_v$ holds an item that only $a_v$ approves,
	and for each $e = \{u,v\} \in E$, the edge agent $a_e$ holds an item that only $a_u$ and $a_v$ approve.
	Thus, each edge agent $a_e$ envies $b$ by one,
	and each vertex agent $a_v$ values its own bundle as valuable as that of each incident edge agent, and more valuable by one than that of each non-incident edge agent.
	We have three additional items $r, r', r^*$ with $\supply(r) = \binom{\ell}{2}$, $\supply(r') = |E| - \binom{\ell}{2}$, and $\supply(r^*) = \ell$.
	All vertex agents $a_v$ approve $r$ and $r^*$, and all edge agents $a_e$ approve $r$ and $r'$.

	Clearly, the reduction can be computed in polynomial time.
	So let us prove its correctness.
	Suppose first that $K \subseteq V$ forms a clique of size $\ell$ in $G$,
	and let $E_k \subseteq E$ be the set of $\binom{\ell}{2}$ edges within that clique.
	Then we create an extension
	by allocating a copy of $r^*$ to each $a_v$ with $v \in K$,
	a copy of $r$ to each $a_e$ with $e \in E_K$,
	and a copy of $r'$ to $a_e$ with $e \notin E_K$.
	Then the edge agents no longer envy $b$.
    Next, consider a vertex agent $a_v$ and an edge agent $a_e$.
	If $e \notin E_K$, then the value of $a_e$'s bundle did not change for $a_v$, so there still is no envy.
	So suppose that $e \in E_K$.
	If $e$ is not incident to $v$, then $a_v$ now values its bundle and that of $a_e$ equally.
	If $e$ is incident to $v$, then $v \in K$, and $a_v$'s valuation for its own bundle and $a_e$'s bundle increased by one.
	Thus, in all, this extension resolves envy.

	Suppose next that $\rho$ resolves envy.
	As each edge agent $a_e$ envies $b$ by one, we need to assign at least one item to each $a_e$ that $a_e$ approves.
	This only leaves the items $r$ and $r'$, and as $\supply(r) + \supply(r') = |E|$, each edge must receive exactly one item.
	Let $E_K$ be the set of edges such that for each $e \in E_k$, agent $a_e$ receives $r$.
	Then each vertex agent $a_v$ where $v$ is incident to an edge $e \in E_K$ must receive a resource $r^*$, otherwise it envies $a_e$.
	As $\supply(r^*) = \ell$ and $|E_K| = \supply(r) = \binom{\ell}{2}$, the edges in $E_K$ must be incident to at most $\ell$ vertices.
	This is only possible if these vertices form a clique of size $\ell$, and $E_K$ are the edges within that clique.
\end{proof}

Moreover, there is little hope for fixed-parameter tractability with the number of agents even if they all have identical valuations.

\begin{proposition}
	\label{prop:ufair-wone-a}
	\ufair{} is \Wone-hard parameterized by~$|A|$,
	even if all numbers in the input are encoded in unary,
	and all agents have identical valuations, even within the initial resource set.
\end{proposition}

\begin{proof}
	We give a reduction from \textsc{Bin Packing} parameterized by the number~$\ell$ of bins:
	Given integers $u_1, u_2, \dots, u_n \in \N$, a number $\ell$ of bins, and a bin size $B$,
	the task is to decide whether there is an assignment of the integers to the bins such that the sum of the integers in one bin is at most $B$.
	This problem is \Wone-hard even if all numbers are encoded in unary and if the sum of all integers equals $\ell B$ \cite{JKMS13}; thus any solution must assign integers of value exactly $B$ to each bin.
	Given a \textsc{Bin Packing} instance, we construct an instance of \ufair{} as follows.
	We have $\ell+1$ agents $a_1, \dots, a_\ell, b$ with identical valuations; we call the valuation function $v$.
	Initially, only $b$ holds an item $p$ with value $v(p) = B$.
	For each $j \in [n]$, $\finR$ contains an item $r_j$ with value $v(r_j) = u_j$.

	Clearly, the construction can be computed in polynomial time.
	We next show the two instances to be equivalent.
	The \textsc{Bin Packing} instance is a \yes-instance
	if and only if there is an assignment $\pi \colon [n] \to [\ell]$ such that each bin receives integers of value exactly to $B$.
	Observe that the extension that allocates $r_j$ to $a_{\pi(j)}$ for each $j \in [n]$ assigns a set of items with value exactly $B$ to each agent $a_i$ and thus resolves all envy with $b$.
	Indeed, an extension resolves envy if and only if it assigns a set of value $B$ to each $a_i$; this proves the equivalence of the two instances.
\end{proof}

We conclude this section with a result that holds for the bounded as well as for the unbounded variant of our problem.
From \cref{prop:ufair-wone-a} we can conclude that parameter number of agents is too small to hope for fixed-parameter tractability.
\cref{prop:np-h-r} demonstrates the same for the number of (different) additional items.
However, when these parameters are combined, the problem becomes fixed-parameter tractable.
This follows from a simple ILP formulation combined with the algorithm by \citet{eisenbrand2020ilp}.

\begin{observation}
	\label{prop:ILP}
	\fair{} 
	is FPT when parameterized by $|A| + |R|$.
\end{observation}

\begin{proof}
	The ILP uses an integer variable $0 \le x_a^r \le \supply(r)$ for each agent $a \in A$ and each item $r \in R$.
	Our formulation contains the following constraints:
	\begin{flalign}
		&& \label{eq:ilp-2}\sum_{a \in A} x_a^r \leq \supply(r) && \text{for all } r \in R
	\end{flalign}%
	\begin{equation}
		\label{eq:ilp-3}
		\begin{multlined}
			v_a(\sigma(a))\! + \!\sum_{r \in R}\! v_a(r) x_a^r \ge{} v_a(\sigma(a'))\! + \!\sum_{r \in R} v_{a'}(r) x_{a'}^r\\
			\text{for all } a, a' \in A.
		\end{multlined}
	\end{equation}
	The first constraint ensures that we add at most $\supply(r)$ copies of $r \in R$,
	while the second constraint ensures that the extended allocation is envy-free.
	If we are given an instance of the bounded variant, we add the constraint
	\begin{equation}
		\label{eq:ilp-1}\sum_{a \in A} \sum_{r \in R} x_a^r\leq k
	\end{equation}
	to ensure that the extension allocates at most $k$ items.

	Fixed-parameter tractability now follows from the fact that our ILP contains $\mathcal O(|A| + |R|)$ variables and $\mathcal O(|A|^2 + |R|))$ constraints
	and from the result by \citet{eisenbrand2020ilp} showing that solving ILPs is fixed-parameter tractable with respect to their number of constraints. 
\end{proof}

\section{When the Budget is Bounded}
\label{sec:bounded-budget}

Finally, we consider the case where we have a bound $k$ on the size of the extension.
We show that, even if the supply is not the bottleneck,
there is little hope for the problem to be fixed-parameter tractable when parameterized by $k$.

\begin{theorem}
	\label{thm:hardness}
	\fair{} with envy-freeness is \Wone-hard with respect to~$k$ even if restricted to binary valuations, if all supplies are $k$, and, initially, there is only one envious agent. 
\end{theorem}

\begin{proof}
	We give a reduction from \textsc{Independent Set} parameterized by solution size~$\ell$,
	which, given a graph $G = (V, E)$ and an integer $\ell$, asks whether there is a set $K \subseteq V$ of at least $\ell$ pairwise non-adjacent vertices.
	Given such an instance, we construct an equivalent instance of \fair{} as follows.
	We have one agent~$a_e$ for each edge~$e \in E$ and one additional ``selection'' agent~$b$.
	The set~$P$ contains~$m\ell$ initially assigned items (where~$m = |E|$) of two different types~$t_1$ and~$t_2$.
	Each agent~$a_e$ initially holds~$\ell-1$ items of type~$t_1$ and one item of type~$t_2$.
	Agent~$b$ approves items of both types and each agent~$a_e$ only approves items of type~$t_2$.
	Note that initially only agent~$b$ is envious (as $b$ does not hold any items).
	The set~$R$ contains one item~$r_v$ for each vertex~$v \in V$.
	The item~$r_v$ is approved by agent~$b$ and all agents~$a_e$ where~$e$ is incident to~$v$ in~$G$.
	We complete the construction by setting~$k = \ell$.
	
	Clearly, the reduction can be computed in polynomial time.
	We next show the two instances to be equivalent.
	Assume first that there is an independent set~$K$ of size~$\ell$ in~$G$.
	Then we create an extension that allocates one copy of each item in $\{r_v \mid v \in K\}$ to $b$.
	As $b$ has positive value for all these items, $b$ does not envy any other agent.
	Each agent~$a_e$ evaluates the bundle of each agent~$a_{e'}$ (including themselves) with~$1$.
	However, since $K$ is an independent set, agent~$b$ holds at most one item that~$a_e$ approves.
	Thus, no agent envies any other agent and the constructed instance is a \yes-instance of \fair.
	
	For the converse, let $\rho$ be an envy-resolving extension of size $\ell$.
	All of these $\ell$ items must be assigned to $b$ as otherwise $b$ envies all other agents.
	If $\rho(b, r_u) + \rho(b, r_v) > 1$ for any edge $e = \{u,v\}$,
	then $a_e$ will envy $b$, as it has positive value for both $r_u$ and $r_v$, but only a value of $1$ for its own bundle.
	Thus, the set of items allocated by $\rho$ form an independent set of size $\ell$ in $G$.
\end{proof}

As the hardness results from \cref{sec:bounded-supply} also hold when we have a size bound $k$,
we cannot hope for fixed-parameter tractability for the number of (different) items.
We are however able to show fixed-parameter tractability with respect to the sum over all item supplies.


We remark that this algorithm implies that \fair{} is in XP with respect to $k$.

\begin{theorem}
	\label{thm:fpt-r-bounded}
	\fair{} can be solved in $\mathcal O(|R|^{\min\{p, k\}} \cdot |A|^2)$ time, where $p = \sum_{r \in R} \supply(r)$ is the sum of the finite supplies.
\end{theorem}

\begin{proof}
	Our algorithm is based on a branching routine
	which is given an instance $(P, R, \supply, A, v, \sigma,k)$
	and an (initially empty) extension $\rho$ that allocates items with finite supply, and proceeds as follows.

    If $\rho$ resolves all envy, return \yes.
    Otherwise, there exists at least one envious agent $a \in A$.
	For each $r \in R$ with $\supply(r) > 0$ (if there is no such item or $k=0$, return \no),
	make a recursive call with input $(P, R, \supply', A, v, \sigma,k-1)$ and $\rho'$,
	where $\supply'$ is a copy of $\supply$ but $\supply'(r) = \supply(r) - 1$
	and $\rho'$ is a copy of $\rho$ but $\rho'(a, r) = \rho(a, r) + 1$.
	Return \yes{} if one of the calls returns \yes{}, and \no{} otherwise.

	Clearly, every extension $\rho$ in any recursive call is valid, i.e., it allocates each item $r \in R$ at most $\supply(r)$ times and breaks as soon as $k$ items have been allocated.
	Thus, if the algorithm computes an extension which is envy-resolving, then we have a \yes-instance at hand.

	For the converse, let $\rho^*$ be an envy-resolving extension of size at most $k$.
	For two extensions $\rho$ and $\rho'$ we write $\rho \le \rho'$ if $\rho(a, r) \le \rho'(a, r)$ for all $a \in A$ and $r \in R$.
	We prove that the algorithm returns \yes{} by induction over the size $k'$ of $\rho$.
	For each $k' \ge 0$, we claim that there is a recursive call that is given an extension $\rho$ of size $k' < k$ such that $\rho \le \rho^*$.
	This statement is clearly true for $k' = 0$.
	So suppose that the statement is true for some fixed $k'$
	and let $\rho \le \rho^*$ be the corresponding extension.
	If the algorithm returns \yes{} within the current recursive call, then we are done.
	Note that specifically, the algorithm returns \yes{} if $\rho(a, r) = \rho^*(a, r)$ for all $a \in A$ and $r \in R$, since $\rho$ is envy-resolving.
	Suppose that the algorithm does not return \yes{} within the current recursive call.
    Note that for each agent $a$, that is envious with respect to $(\sigma,\rho)$,
    and thus particularly the agent that is chosen by the algorithm within our call,
    the solution $\rho^*$ must assign at least one (additional) item from $R$ to that agent;
    thus, $\rho(a, r) < \rho^*(a, r)$ for some item $r \in R$.
	This also implies that $\supply(r) > 0$ in the current recursive call.
	Then, the algorithm makes a call that is given an extension $\rho'$ that is a copy of $\rho$, but $\rho'(a^*, r) = \rho(a^*, r) + 1$.
	Note that $\rho'$ allocates $k'+1\le k$ finite supply items and $\rho' \le \rho^*$, proving the induction, and thus the correctness of the algorithm.

    As for the running time, checking whether the current extension resolves envy takes $\mathcal O(|A|)$ time.
    In each recursive call, one item is assigned and $k$ is decreased by one.
    Thus, the depth of the branching tree is at most $\min\{p, k\}$.
    Moreover, as in each call, the algorithm makes at most $|R|$ recursive calls, the overall running time is as claimed.
\end{proof}

\section{Conclusion}
\label{sec:conclusion}
We initiated the study of mitigating fairness of an initial allocation by adding goods,
restricting ourselves to the notion of envy-freeness and additive valuations.
In our parameterized complexity analysis, we leave open whether \fair{} is NP-hard or polynomial-time solvable for a constant number of agents.
Further, can \cref{thm:fpt-r-bounded} be lifted to the setting where the budget is unlimited and there are both finite and infinite supplies?
(The parameter would then become the sum of the finite supplies.)
When adapting the branching routine to only assign finite supply items and calling the polynomial-time algorithm from \cref{thm:poly} to decide whether the current allocation can be extended with the infinite supply items, one would need to derive a hint which agent is guaranteed to require an item with finite supply.
Finally, does \cref{prop:np-h-r} also hold for two additional items?

Our problem can also be studied for other fairness notions such as EF1 \cite{DBLP:conf/bqgt/Budish10} or EFX \cite{DBLP:journals/teco/CaragiannisKMPS19}, but also the many fair share notions \cite{steinhaus1948fair,budish2011mms,babichenko2024quantile}.
As for EF1 and EFX, we believe that our hardness results can be easily transferred.

Finally, investigating approximation guarantees, e.g.\ when minimizing
the number of added items could be an intriguing research direction.
Similarly, it seems attractive to minimize (absolute or relative) envy approximately
in our setting.
Since the latter approximation is known to be mostly intractable when allocations
are computed from scratch~\cite{Lipton2024approxenvy}, it might be worth
to combine this with parameterized approaches or structural restrictions.
\clearpage 

\subsection*{Acknowledgements}
This research project started at the annual research retreat of the
Algorithmics and Computational Complexity (AKT) group of TU
Berlin, held in Darlingerode, Germany, September 17--22, 2023.
Matthias Bentert acknowledges his support by the European Research Council (ERC) under the European Union’s Horizon 2020 research and innovation programme (grant agreement No.\ 819416).
Eva Deltl acknowledges her support by the Deutsche Forschungsgesellschaft (German Research Foundation, DFG), project COMSOC-MPMS, (grant agreement No.\ 465371386).
Pallavi Jain acknowledges her support from SERB-SUPRA grant number SPR/2021/000860 and IITJ Seed Grant grant I/SEED/PJ/20210119.
Robert Bredereck, Eva Deltl and Pallavi Jain were supported by a DST-DAAD travel grant (DAAD grant agreement No.\ 57683502).

\bibliographystyle{named} 
\bibliography{references_clean}

\end{document}